\newif\ifpaper
\papertrue      
\ifpaper
\documentclass{IEEEtran}  

\newcommand{\FigureWidth}{8.5cm} 
\else
\documentclass[onecolumn,draftclsnofoot,12pt]{IEEEtran}  
\usepackage{setspace}
\newcommand{\FigureWidth}{17cm}  
\fi
\ifpaper
\newcommand{\TwoOneColumnAlternate}[2]{#1}
\else
\newcommand{\TwoOneColumnAlternate}[2]{#2}
\fi

\usepackage[font=small,labelfont=bf]{caption}
\usepackage{subfig}
\usepackage{color}
\usepackage[final]{graphicx}
\usepackage{bm}
\usepackage{amssymb,amsmath,amsthm}
\usepackage[normalem]{ulem}
\usepackage{paralist}
\newtheorem{lemma}{Lemma}
\usepackage{lipsum}

\newcommand\blfootnote[1]{%
  \begingroup
  \renewcommand\thefootnote{}\footnote{#1}%
  \addtocounter{footnote}{-1}%
  \endgroup
}
\begin{document}

\newcommand{\argmax}{\mathop{\mathrm{argmax}}} 
\newcommand{\argmin}{\mathop{\mathrm{argmin}}}
\newcommand{\diag}{\mathop{\mathrm{diag}}}
\newcommand{\tr}{\mathop{\mathrm{tr}}}
\newcommand{\ZF}{\mathop{\mathrm{ZF}}}
\newcommand{\tx}{\mathop{\mathrm{tx}}}
\title{ Optimal and Suboptimal Routing   Based on Partial CSI in Random Ad-hoc Networks}
 \author{\IEEEauthorblockN{Yiftach Richter and Itsik Bergel}\\
}

\maketitle
\begin{abstract}
In this paper we consider routing in random 
wireless-adhoc-networks (WANETs), where each node is equipped with a single antenna. 
Our analysis uses a proper model of the physical layer together with an abstraction of higher
communication layers. 
We assume that the nodes
are distributed according to a Poisson-point-process and consider
routing schemes that select the next relay based on the
geographical locations, the channel gains of its neighbor nodes and the statistical characterization of all other nodes.  While many routing problems
are formulated as optimization problems, the optimal distributed
solution is rarely accessible. In this work, we present the exact
optimal solution for the  scenario analyzed. The optimal routing is
given as a maximization of a routing metric which  depends solely
on the known partial channel state information (CSI) and includes an expectation
with respect to  the interference statistics. The optimal routing scheme is
important because it gives an upper bound on the performance of
any other routing scheme. We
also present sub-optimal routing schemes that only use  part of
the available knowledge and require much lower computational
complexity. 
Numerical results demonstrate
that the performance of the low complexity schemes is close to
optimal and outperforms other tested routing schemes.
\end{abstract}





\section{Introduction}
\blfootnote{
The authors are with the Faculty of Engineering, Bar-Ilan University, 52900
Ramat-Gan, Israel (e-mail: richtey@biu.ac.il; Itsik.Bergel@biu.ac.il).
\par
Preliminary results were published in conference proceedings: in the IEEE $16$-th International Workshop on Signal Processing Advances in Wireless Communications ({SPAWC}) $2015$, and in the IEEE $28$-th Convention of Electrical and Electronics Engineers in Israel ({IEEEI}) $2014$.

}
Wireless ad-hoc networks (WANETs) have become very popular due to their  flexibility and scalability. In order to ensure reliable communication without 
any infrastructure, WANETs are based on  cooperation among nodes. In such networks, 
direct communication is usually undesirable, primarily because of the high energy consumption and the  interference on neighbor nodes. Instead, WANETs commonly employ  multihop routing in which  messages are delivered from sources to destinations using intermediate nodes that operate as relays.
\par
{{ 
Over the years, the design and analysis of routing algorithms has  attracted a significant research attention (see for example \cite{royer1999review,abolhasan2004review} and references therein).
Considering multihop routing optimization, the routing should
take into account both  the networking parameters together with several cross-layer parameters (e.g., channel states,
transmission powers, 
transmission probability, etc).
In particular, most wireless networks operate over time-varying channels, where each node experiences varying channel gains between itself and its neighbors.
The performance of WANETs heavily depend on these channel states. 
\par
Only  few works have analyzed the  combination  of physical layer parameters (channel status, transmission power etc) and networking parameters. While there have been several suggestions in this direction, so far there is much uncertainty about the optimal way to weigh the different parameters in the routing decision.
Hence, this has remained one of the main open issues in routing for wireless networks.  
\par
It is also important to define what  physical measurements  are available for each node. In general, it is not realistic to assume that each node has full channel state information (CSI); I.e., knowledge of  
the locations of all the other nodes and the channel gains between each two nodes. In this paper we assume a more practical `partial' CSI scenario, in which, each node can acquire the locations of its neighboring nodes and the channel gains between
itself and each of its neighbors. 
}}
\par
{{
Analysis of random WANETs allows
a study of the behavior of a WANET over many spatial realizations
in which the nodes are placed according to some
probability distribution.}}
 Perhaps the most  popular modeling of  random WANETs uses    Poisson point processes (PPPs) (e.g. \cite{baccelli2010stochastic,haenggi2009stochastic}).
To further simplify the analysis, most works  using PPP modeling for routing have considered \textit{geographic routing} (e.g., \cite{baccelli2008performance}\nocite{nardelli2012efficiency}--\cite{torrieri2013multihop}). In geographic routing,  the
next hop is selected on the basis of  knowledge about the geographical locations of
the potential relays, and the routing  decision at each hop is independent of all other  decisions.   
\par
{{Taking into consideration the instantaneous CSI knowledge, the performance of the routing algorithms  can be significantly improved.
This type of algorithm is known as \textit{opportunistic relaying}. This paper considers the specific case of geographic routing together with opportunistic relaying. This is a challenging case, as the  location information and channel information are hard to combine.}}
\par
Weber et al. \cite{weber2008longest} suggested using  opportunistic relaying by first identifying a set of relays for which the received power
is larger than a certain threshold, and then selecting the farthest
node in the set as the next relay. Assuming that each node has
many messages in its buffer, the algorithm searches for a
message in the buffer that can be routed via the selected relay
node. 
Hence, all routing in the network is performed over good links.
 Zanella et al. \cite{zanella2013impact}
 investigated the impact of various routing algorithms on the
generated interference, as well as the trade-off between the
average number of hops and the generated interference. They
compared  simple geographic routing algorithms and
opportunistic relaying algorithms and proved the superiority
of opportunistic relaying algorithms. 
\par
In traditional transmission schemes, all packets have the same length and use the same code rate. Thus,  successful decoding can only take place with high probability  if the received  signal quality (i.e., the signal to  interference plus noise  ratio - SINR) is above a threshold.  If the receiver cannot decode the message,  the transmission is said to be in `outage'.
In this  case, all the received bits are discarded and the transmitter retransmits the same data until it is decoded.  Obviously, the discarding of undetected data is a waste of resources, and the network throughput is not optimal. Furthermore, to ensure a low outage probability, the transmitters must use a relatively low code rate, which leads to further reductions in throughput.
\par
In this work we consider the ergodic rate, which is the maximal achievable rate in such links and can be significantly higher than the outage rate. The ergodic rate  is given by the mutual information between the transmitted signal and the received signal (given the interferers' activity). The ergodic rate is always higher than the outage rate, but in some cases it may require a large delay.
\par
The achievability  of the ergodic rate in systems with time-diversity or frequency diversity is discussed in detail in the literature (e.g., 
\cite{biglieri1998fading}\nocite{ERD_LB_2013_Yaniv}--\cite{george2017ergodic}). With a sufficient delay, the ergodic rate can be achieved with no outage. Rajanna et al. \cite{rajanna2015performance} showed that the ergodic rate can be approached  with limited delay, by allowing a small outage probability.
\par
The achievability of the ergodic rate  can be illustrated through the example of the Hybrid Automatic Repeat Request (HARQ)  (e.g., 
\cite{larmo2009lte,dahlman20103g}). In HARQ, when the receiver cannot decode successfully, it stores the received signal in memory and waits for  transmissions of additional code bits from the transmitter.  Note that even if the received
signal has led to a detection failure, it still contains useful information about
the transmitted packet which will be useful later for detection. The transmission of additional code bits continues until the receiver is able to decode the packet. Hence, each message will eventually be decoded, and  the data are transferred at a rate which is very close to the maximal rate for this link (the ergodic rate).
\par
The HARQ protocol is commonly used in modern communication systems such as HSPA and LTE. It makes it possible to improve the communication reliability at rates that can lead to a significant outage probability. Note that while HARQ can work with no outage at all, its actual   transmission rate for each package is not determined in advance and depends on the channel and network conditions. If the HARQ protocol is tuned to start with a high initial rate, and adds a small number of code bits on each try, the average data rate over all packages will be very close to the ergodic rate.
\par   
Conveniently, the ergodic rate is also easier for analysis.
The network performance measure  associated with the ergodic rate is termed the Asymptotic-Density-of-Rate-Progress (ADORP) \cite{Richter2014_IEEEI_SO}. The ADORP measures the average density of the product of the rate and the distance of each transmission. This gives a good indication of the capability of the network to deliver messages from sources to destinations (see more details in \cite{Richter2014_IEEEI_SO}). The ADORP is equivalent to the transport-capacity
\cite{andrews2010random} except for the use of the ergodic rate instead of the outage rate.
\par
Although \cite{Richter2014_IEEEI_SO} aimed to derive a routing metric that optimizes the network throughput in single antenna WANETs, the actual derivation ignored some of the data that are known in advance at the transmitting node. 
In this paper,  we present    the routing metric that  maximizes the WANET throughput. 
\par
All previous works have presented a routing metric, and
then evaluated its performance. In this paper we take a more direct approach, and  derive the exact routing function that maximizes the ADORP for single antenna WANETs. As this routing scheme is optimal with respect to the (known) network statistics, we term it  Statistically-Optimal (SO) routing. We also present three sub-optimal, low-complexity routing schemes that can be evaluated in a closed form. Finally, we show that the routing scheme of \cite{Richter2014_IEEEI_SO} (termed here narrow-bound-optimal (NBO) routing) can be viewed as an approximation of the SO routing \cite{richter2015optimal}. Thus, we can compare the performance of all the routing schemes; we show  that the NBO routing scheme is very close to the optimal routing (given only local information at each transmitting node).
\par
The simplicity of the NBO scheme, together with its near optimality make it a good candidate for routing in practical single antenna WANETs. Furthermore, the structure of the derived routing scheme is based on the evaluation of a routing metric for each candidate relay node. This structure enables easy integration with other routing schemes that take  additional constraints into account (e.g., traffic loads,
delays, user requirements, etc.').

\par
To demonstrate the superiority of these routing schemes, we compare their performance to the performance of commonly-used routing schemes proposed in previous works (e.g. \cite{haenggi2005routing,nardelli2012efficiency}). 
\par
{{ 
 This paper takes a significant step towards optimal cross layer routing by considering the  joint optimization of the routing decisions together with the physical layer. Using this approach, the presented results allow much better understanding on the optimal relation between physical layer parameters and networking parameters. }}
\par
{{ 
The main contributions of this  paper is the presentation, for the first time, of an optimal cross layer routing in  WANETs. The optimal routing uses all the available knowledge at each node.
The paper also  presents a suboptimal scheme that requires only a low computational complexity and uses only part of the available CSI. The performance gap between the optimal scheme and this low complexity scheme is shown to be negligible. 
}}
\par
The rest of this paper is organized as follows. Section \ref{Section: System Model}
describes the structure of the WANET. Section 
\ref{Section: Routing SISO WANET} 
 presents the novel routing schemes. Section 
\ref{Section: Numerical Results} analyzes the numerical results, and Section \ref{Section: Conclusions} presents the conclusions.

\textit{Notations}: 
Throughout this paper, matrices and vectors are denoted
by boldface symbols. The conjugate of a complex number
is marked by $(\cdot)^*$, and $(\cdot)^T,(\cdot)^H$ denotes the transpose and
the conjugate transpose of a matrix, respectively. The expectation
and probability of a random variable (r.v.) are denoted by $\mathbb E(\cdot)$ and $\mathbb P(\cdot)$, respectively. $\mathbf I_N$ is the $N \times N$  identity matrix, and $\|\mathbf x\|$ is the
Frobenius norm of the vector $\mathbf x$.

\section{System Model\label{Section: System Model}}
We assume a decentralized WANET over an infinite area, where each node is equipped with a single   antenna.
The locations of the nodes are modeled by a homogeneous
Poisson point process (PPP),  
$\Phi$, with density $\lambda$ (i.e., the
number of nodes in any area of size $A$ has a Poisson
distribution with a mean of $\lambda A$).



\subsection{Medium-Access-Control (MAC)}
We use the common slotted ALOHA medium-access
(MAC) model (e.g., \cite{arnbak1987capacity,baccelli2006aloha,haenggi2009stochastic,baccelli2010stochastic}) 
where each node chooses independently to  transmit with probability $p_{\mathrm{tx}}$ or a  listening receiving node with probability $(1-p_{\mathrm{tx}})$. 
These independent transmission decisions give a simple  yet robust protocol that does not require coordination between the nodes.
Thus, the locations of the transmitting nodes can be represented by the PPP $\Phi_T$ with a
density of $\lambda p_{\mathrm{tx}}$.
 The MAC decisions are taken locally and independently. Thus, each node does not know which of the  other nodes  are scheduled to transmit. 

\subsection{Physical Layer (PHY)}
 The received signal in the $i$-th receiving node is given by
\begin{IEEEeqnarray}{rCl}\label{e: received signal main}
y_i=
\sum \limits_{\substack{j\in \mathbb N}}
\sqrt{\rho}\cdot
r_{i,j}^{-\frac{\alpha}{2}}h_{i,j}
z_{j}
+  v_i
\end{IEEEeqnarray}
where 
 $r_{i,j}$ and $h_{i,j}$ are the distance and the channel gain between the $j$-th transmitting node and the
$i$-th receiving node, respectively. We assume throughout that all channel gains,  
$h_{i,j},$ are composed of statistically independent and identically
distributed (i.i.d.) complex Gaussian random variables with
zero mean and unit variance.  The data symbol of the $j$-th transmitting node is represented by the scalar 
$z_j$, where all data symbols are i.i.d. standard complex Gaussian random variables, $z_j\sim\mathcal C \mathcal N(0,1)$. The thermal noise at the $i$-th receiving node,  $v_i$,   is assumed to have a complex Gaussian distribution
with zero mean and $\mathbb E\{v_i v_i^H\}=\sigma_v^2$.   The path loss exponent   is denoted by $\alpha$ and satisfies $\alpha>2$.
  We consider nodes with an identical transmission power, $\rho$.

We define  
 \begin{IEEEeqnarray}{rCl}\label{e: W_i_j}
W_{i,j}
\triangleq
|h_{i,j}|^2
\end{IEEEeqnarray}
to denote the fading variable between  the $j$-th transmitting node and  the $i$-th receiving node. 
It may be noted that $W_{i,j}\sim\text{Exp}(1)$, $\forall i,j$. Considering the decoding of the data-symbol  from transmitting node $j$ at receiving node $i$, the desired signal power is
\begin{IEEEeqnarray}{rCl}\label{e: power of the k-th data stream}
S_{i,j}
=
\rho
r_{i,j}^{-\alpha}
W_{i,j}
\end{IEEEeqnarray}
and the power of the interference  is 
\begin{IEEEeqnarray}{rCl}\label{e: interference of the k-th data stream}
J_{i,j}
&=&
\rho\sum \limits_{\substack{\ell\ne j}}
r_{i,\ell}^{-\alpha}
W_{i,\ell}.
\end{IEEEeqnarray}
It may be noted as well that  all fading variables (i.e., 
$W_{i,j}$) are i.i.d.

The signal-to-interference-plus-noise-ratio (SINR) of the $j$-th transmitting node that is detected at receiving node $i$ is given by 
\begin{IEEEeqnarray}{rCl}\label{e: SINR}
\text{SINR}_{i,j}=
\frac{S_{i,j}}{J_{i,j}+\sigma_v^2}.
\end{IEEEeqnarray}
Assuming a near optimal coding scheme and a  long enough code word, the rate contribution per slot from transmitting node $j$ to receiving node $i$  is given by
\begin{IEEEeqnarray}{rCl}\label{e: achievable rate over k-th data stream}
R_{i,j}=
B\cdot\log_2\left(1+
\text{SINR}_{i,j}
\right)
\end{IEEEeqnarray}
where $B$ is the channel bandwidth (in Hertz). 
\par
Thus if two transmitters decide to transmit to the same receiver node, the two packets will be received in an interference-limited manner; i.e., the receiver will decode each message, while considering the other message as noise.  Hence, we can consider each message separately, and the interference term for each message includes all other transmitted messages, regardless of their destinations.
\subsection{Routing Mechanism \label{subsection: routing mechanism}} 
Each message has an origin (source) node and a destination node. The routing algorithm needs to forward messages from their sources to their destinations  through nodes which serve as relays. 
As stated above, we focus on geographical routing together with opportunistic relaying \cite{weber2008longest}. Specifically, the decisions on the next-hop are based on the locations of the nodes, and the selections are independent among nodes. In opportunistic relaying, a transmitter first selects the next relay  based on channel states and relay locations. Afterward, the transmitter searches in its buffer
for the message that gains most from the use of this relay. Note that if the message buffer is long enough, the destination of the selected message will typically be very close to the line that extends from the source to the selected relay. Thus, in opportunistic relaying, all routes will be almost  straight lines, and all transmissions will be at good channel conditions.
\par
We focus on the case where each transmitting node only has  knowledge of the  locations  of its  neighbor nodes, and the channel gain to each neighbor.  We define  two nodes to be neighbors if their distance is at most $r_{\rm{A}}$. Thus, the neighborhood of node $j$ is the set of indices:
\begin{IEEEeqnarray}{rCl}\label{e: N, set of nodes within routing zone}
\mathcal N_j\triangleq\left\{
i: \| \mathbf r_{i}-\mathbf r_{j} \| \le r_{\rm{A}}
\right\}
\end{IEEEeqnarray}
where the vector $\mathbf r_i$ contains the coordinates of the $i$-th node.
We also use the term \textit{routing zone} to describe the area of all potential neighbors; i.e., the
circular area of radius  $r_{\rm{A}}$ that is centered at each node.
 
 The available knowledge of the $j$-th transmitting node is denoted as $\mathcal M_j$. It contains  the channel gains and the locations of the nodes within its routing zone. Using the notations given above, this knowledge can be written mathematically as the set:
\begin{IEEEeqnarray}{rCl}\label{e: set M, local knowledge}
\mathcal M_j \triangleq \Big\{
(\mathbf r_i,h_{i,j}):\forall i \in \mathcal N_j
\Big\}.
\end{IEEEeqnarray}

This local knowledge can be achieved in various ways. For instance, one can assume that each node can acquire its location using a GPS receiver. Each node also shares its location with its neighbors jointly with the transmission of data  messages. After several slots, each node can know the  locations of all nodes in its neighborhood.
The CSI to each of the neighbors  is usually obtained by pilot-based channel estimation at the receivers (e.g., \cite{yoo2006capacity, baltersee2001information}). WANETs typically use time division duplex (TDD); i.e., they transmit and receive on the same frequency. Thus, the CSI for the transmitters is typically obtained by using the channel reciprocity (e.g., 
\cite{guillaud2005practical,liang2016fdd}).
\par
We denote the routing selection   of transmitting node $j$ by the function $f(\mathcal M_j)$, which here will be termed  the  \textit{routing function}. The  routing function receives  the available knowledge, $\mathcal{M}_j$, as input; i.e., the locations of all nodes in the routing zone and the channel gain from each   of these neighbors. 
The function output is the routing selection, i.e., $f(\mathcal M_j)\in\mathcal N_j$ is the index of the selected relay for the next hop. 
\subsection{Routing Performance} 
While the routing mechanism aims to deliver messages from their sources to their destinations (through several hops), our goal here is to measure  performance through the analysis of a single time slot.

Since we want to analyze the maximal network performance, we assume that the messages are generated in a homogenous manner in all nodes of the network. We also assume that the message generation rate is high enough  that the message buffers of all the nodes are rarely empty. We also assume that none of the nodes becomes a network bottleneck, so that the data flow in the network is homogenous. For this assumption to hold, we  must also detail our assumptions on network mobility.

\par
We first need to distinguish between three popular mobility models. Each model differs in terms of the relationship between the node mobility and the message transfer time 
(i.e., the time that it takes for a message to get from its source to its destination). In the extreme \textit{very-fast-mobility} model, the node mobility is considered to be much higher than the message transfer time. Furthermore, a node may keep the message in its buffer and transmit it only when the node is close to the destination \cite{grossglauser2001mobility}. Clearly, the very-fast-mobility model is analogous to a very loose delay constraint.

In the second model, commonly termed the \textit{static} model (e.g. \cite{gong2013local,blaszczyszyn2015random}), the node mobility is very low and the topology of the network changes very slowly compared to the message transfer time (i.e., the delay constraint is quite short). In this static model, the specific network topology becomes crucial, because some nodes will need to relay more messages than others. 
These nodes may become network bottlenecks as a result of high message traffic. (For example, consider the case where a node is located between two groups of
nodes, so that all message traffic between these groups must pass through this node. This node is a
network bottleneck, and will be required to handle a larger communication load than its neighbors). In this (quite practical) model, the total network throughput depends critically on the performance of the bottlenecks.

In the third model, the popular \textit{fast-mobility} model (e.g. \cite{gong2013local}), the node mobility is considered to be smaller than the message transfer time. However, the node mobility is sufficient to change the network structure just enough that none of the nodes will be a network bottleneck for a long time. Thus, over a long enough period of time, all the nodes will experience all possible network neighborhoods, and the network can be considered  homogenous. In this work we focus on this fast mobility model, as it relatively easy to analyze, and can give at least an upper bound on the achievable performance in real life networks.

We also assume that each node has a very long buffer.
Under these assumptions, it is reasonable to assume that  small movements in the network will cause enough changes to the network topology so that over a long enough observation time,  the network can be considered  homogenous. Thus,  the routing performance can be characterized by the analysis of a typical transmitting node located at the origin and transmitting at a typical time slot.

 We refer to the probe transmitting node as transmitting node $0$ (which can either transmit a new message, or relay
a message received from another user). The available knowledge of the probe transmitting node is given by $\mathcal M_0$, and the next hop selection is given by  $f(\mathcal M_0)\in\mathcal N_0$. 
The Asymptotic-Density-of-Rate-Progress (ADORP) performance metric is given by  \cite{Richter2014_IEEEI_SO}: 
\begin{IEEEeqnarray}{rCl}\label{e: ADORP SISO definition}
\bar D\left(f(\cdot)\right)
 \triangleq
\lambda  p_{\mathrm{tx}}(1-p_{\mathrm{tx}}) 
\mathbb E\left\{r_{f(\mathcal M_0),0} R_{f(\mathcal M_0),0}\right\}
\end{IEEEeqnarray}
where $\lambda p_{\mathrm{tx}}$ is the density of the  active transmitting nodes and $(1-p_{\mathrm{tx}})$ is the probability that the selected relay is indeed listening (recall that that each transmitting node does not know which other nodes are scheduled to transmit in the current time slot). In words, the ADORP is the density of good transmissions multiplied by distance-rate product; i.e., the  distance between the typical transmitting node-relay pair multiplied by its achievable rate. 

The ADORP provides a convenient way to evaluate the contribution of a single transmission to the total network throughput. Both high data rate and long transmission distance have the same effect on faster transfer of messages in multi-hop routing.  

It may be noted that unlike certain other works (e.g. \cite{ravindran2010optimized,li2010selection}) we simply consider the hop-length and not the progress toward a specific destination because of our long delay assumption. Using the long buffer at each node, we can employ an opportunistic relaying scheme \cite{weber2008longest} where a transmitting node first selects a preferred relay and afterwards finds a message in the buffer that should go in the chosen direction.  Obviously, if the
buffer size is very large, each relay is located on the line to
some desired destination. Therefore, we  assume that in this setup, the routes will be very close to the straight lines between sources and destinations. 


The ADORP, (\ref{e: ADORP SISO definition}),  can be explicitly written as  
\begin{IEEEeqnarray}{rCl}\label{e: ADORP SISO redefinition}
\bar D\left(f(\cdot)\right)
& \triangleq&
\lambda B p_{\mathrm{tx}}(1-p_{\mathrm{tx}})
\TwoOneColumnAlternate{\\ \notag &&}{}   
\cdot  \mathbb E\left\{r_{f(\mathcal M_0),0}
\log_2\left(1+\frac{S_{f(\mathcal M_0),0}}{J_{f(\mathcal M_0),0}+\sigma_v^2}\right)\right\}.
\end{IEEEeqnarray}
The expectation in \eqref{e: ADORP SISO redefinition} considers the routing function, the available knowledge, the desired signal and the aggregated interference.

\section{Routing Schemes  Based on Local Knowledge
\label{Section: Routing SISO WANET}} 
In this section we derive the optimal routing function, and  present three suboptimal routing functions that require much lower computational complexity.

\subsection{Statistically-Optimal (SO) Routing}

Using the law of total expectation and conditioning on the known local information,  the ADORP performance metric,  
 \eqref{e: ADORP SISO redefinition}, can be written as
\begin{IEEEeqnarray}{rCl}\label{e: ADORP law total expectation}
\bar D(f(\cdot))=\lambda B p_{\mathrm{tx}}(1-p_{\mathrm{tx}}) \mathbb E_{\mathcal M_0} \Big\{ r_{i,0}G\left(f(\mathcal M_0),\mathcal M_0\right) \Big\}
\IEEEeqnarraynumspace
\end{IEEEeqnarray}
where 
\begin{IEEEeqnarray}{rCl}\label{e: G(i,k)}
G(i,\mathcal M_0)\triangleq\mathbb E_{J|\mathcal M_0}\left\{\log_2\left(1+\frac{S_{i,0}}{J_{i,0}+\sigma_v^2}\right)\Big|\mathcal M_0\right\}.
\IEEEeqnarraynumspace
\end{IEEEeqnarray}

It may be noted that  $S_{i,0}\in \mathcal M_0$, and hence $S_{i,0}$ is known when $\mathcal M_0$ is known. Thus, the expectation in \eqref{e: G(i,k)} is taken only with respect to $J_{f(\mathcal M_0),0}$. 

Crucially, the
routing decisions have no effect on the interference. This lack of effect  occurs because  the MAC
decisions (i.e., when to transmit) are independent of the routing decisions.
Hence, the optimization of the routing function for the probe receiving node is independent of the routing functions of all other transmitting nodes and can be solved directly from \eqref{e: G(i,k)}.

The optimal routing function can be easily derived
by maximizing the internal expectation of \eqref{e: ADORP law total expectation}.  Specifically, the routing function that  optimizes \eqref{e: ADORP law total expectation}  is 
\begin{IEEEeqnarray}{rCl}\label{e: opt metric SISO first}
f_{\rm{SO}}(\mathcal M_0)
&=&
\argmax_{\substack{
i\in \mathcal N_0 }}
r_{i,0} \cdot G(i,\mathcal M_0)
\\
&=&
\argmax_{\substack{
i\in \mathcal N_0 }}
m_{_{\rm{SO}}}(i,\mathcal M_0).
\notag
\end{IEEEeqnarray}
In other words, the routing function can be written as the solution to an optimization problem where the solution is the index of the node that maximizes a routing metric. For the optimal solution to 
\eqref{e: ADORP law total expectation}, the routing metric is the expectation over the throughput times the distance to the candidate relay:
\begin{IEEEeqnarray}{rCl}\label{e: SO routing metric}
m_{_{\rm{SO}}}(i,\mathcal M_j)\triangleq r_{i,j} \cdot G(i,\mathcal M_j).
\end{IEEEeqnarray}
To reiterate, the  routing metric $m_{_{\rm{SO}}}(i,\mathcal{M}_j)\in \mathbb{R}^+$ is the score for each candidate node, and the routing function, $f(\mathcal M_j)\in\mathcal N_j$ is the index of the selected node (with the highest score).
\par
The optimization part of this problem is quite simple since we only need to evaluate the metric for each node in the routing zone, and choose the best node (typically, the number of nodes in the routing zone will be quite small). On the other hand, the evaluation of the routing metric can be very demanding, which may make this optimal scheme unpractical. This evaluation depends on the conditional distribution of 
$J_{i,0}|\mathcal M_0$. In the numerical results section below, we evaluate this expectation by implementing Monte Carlo simulations for the given local knowledge, $\mathcal M_0$.

The SO routing is too
complicated to be used in practical networks. However, this does not
reduce the importance of the SO routing. Because routing is a complicated task, it is rare to be able to characterize its optimal performance. In the scenario  below, this is the first result that presents an optimal routing scheme and allows the
evaluation of the optimal performance. This optimal performance is an important benchmark for any other scheme, and it serves to quantify the `distance' of each scheme from optimality. In the following, we use the SO performance as a reference when we derive sub-optimal routing schemes with reduced complexity.
\par
Apart for local knowledge, the evaluation of the expectation requires knowledge of the network parameters (e.g., the node density, the transmission probability and the path-loss exponent). These parameters can be hardcoded  during network deployment or estimated at each node (see for example 
\cite{srinivasa2009path}).
\par
One approach (which is also used in the numerical results section below) is to evaluate
this expectation through Monte Carlo (MC) simulations. In this approach, at each MC step (a realization of a random network around the transmitter), we use a different distribution to  model the nodes inside and outside  the routing zones. For the nodes that are located \textit{within} the routing zone, given the local knowledge, the transmitter knows the node locations but not their activity (recall that the transmitter does not know which of those nodes  is scheduled to transmit). Thus, in each MC step, each of these nodes has a probability of $p_{\tx}$ to transmit, independent of its neighbors. The nodes that are located \textit{outside} the routing zone are different because their locations are not known to the transmitter. Thus, the MC uses the PPP modeling, taking into account the node density and the transmission probability.
\par 
In the next subsections, we present three alternative routing functions that decrease the complexity. Each of the schemes is based on  
 a simplified optimization. Using the performance of the SO scheme as reference, we  show that these suboptimal schemes achieve performances that are very close to optimal.

\subsection{Bound-Optimal (BO) Routing}
\begin{figure}[t]
    \center
    \includegraphics[scale=0.5,trim={5cm 4.3cm 5cm 4.1cm},clip=true]{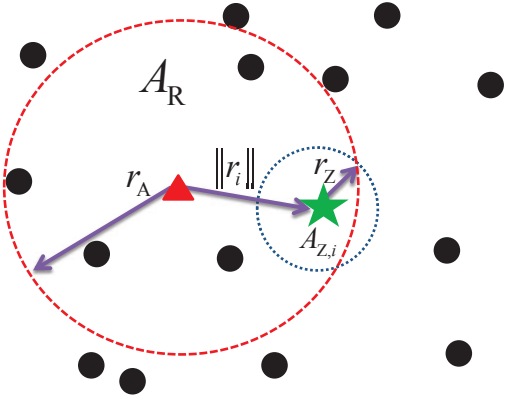}
    \caption{Local neighborhood of the probe transmitting node: each dot represents a node in the network. The triangle is the probe transmitting node and the star is the tested relay (node $i$ in this example). The dashed circle represents the routing zone, $A_{\mathrm{R}}$, with radius $r_{\mathrm{A}}$ centered at the probe transmitting node. The distance between
the probe transmitting node and the tested relay is $\|\mathbf r_i\|$. The threshold zone of node $i$, $A_{\mathrm{Z},i}$, is marked by the small dotted circle with radius $r_{\mathrm{Z}}$ around the tested relay. These circles are used in the proof of Lemma \ref{Theorem: lemma G(i,M) lower bound} as the maximal distance, which requires special attention.
{{The routing chooses the next-hop that has the maximal routing metric, based on the  routing scheme. The difference between schemes is the complexity of the evaluation of the  routing metric and its quality (see also Table \ref{Table: Routing Characteristic}).
 }}
}
    \label{fig: Routing Zone General}
\end{figure}
To derive this low complexity routing scheme, we replace the optimization of 
\eqref{e: opt metric SISO first} by an optimization of the following lower bound on \eqref{e: opt metric SISO first}. This lower bound  is a slightly modified version of the lower bound proposed by  George et al. \cite{ERD_LB_2013_Yaniv,ERD_UB_2014_Yaniv}.

\begin{lemma}
\label{Theorem: lemma G(i,M) lower bound}
For each node within the routing zone ($i\in\mathcal N_0$), denote the disk centered around the node with a radius of $r_{\rm{Z}}= \left(\tfrac{\alpha-2}{\alpha\pi \lambda p_{\mathrm{tx}}}\right)^{1/2}$ as its threshold zone,  $A_{\rm{Z},i}$ (see Fig. \ref{fig: Routing Zone General}), and denote by $N_{{\rm{Z}},i}$ the number of nodes in the intersection between the  threshold zone,  $A_{\rm{Z},i}$, and the  routing zone (with radius $r_{\rm{A}})$. The function $G(i,M)$ in \eqref{e: G(i,k)} can be  lower bounded  using only the known local knowledge, $\mathcal M_0$,  and the system parameters by:
\begin{IEEEeqnarray}{rCl}\label{e: lemma G(f,K) LB}
G(i,\mathcal M_0) \ge
p_{\rm{Z}}(i,\mathcal M_0) 
 \cdot r_{i,0} \log_2
\left(1+\frac{ S_{i,0}}{\bar J_1^i +\bar J_2^i +\sigma_v^2}
\right)
\IEEEeqnarraynumspace
\end{IEEEeqnarray}
where 
\begin{IEEEeqnarray}{rCl}\label{e: p_Z}
p_{\rm{Z}}(i,\mathcal M_0) &&
\TwoOneColumnAlternate{\\\quad \notag &&}{} 
=
\begin{cases}(1-p_{\mathrm{tx}})^{N_{{\rm{Z}},i}}, & {\rm{if  }}\,\  \|\mathbf r_{i}\|+r_{\rm{Z}}<r_{\rm{A}} \\
e^{-\lambda p_{\mathrm{tx}} B_{{\rm{T}},i}}(1-p_{\mathrm{tx}})^{N_{{\rm{Z}},i}}
, & \rm{otherwise} \\
\end{cases}
\end{IEEEeqnarray}
 is the probability that there is no interfering transmitting node within the threshold zone, 
\begin{IEEEeqnarray}{rCl}\label{e: B1 by integral}
B_{{\rm{T}},i}
&=&
r_{\rm{Z}}^2 \tan^{-1}\Big(\tfrac{x_0-\|\mathbf  r_{i}\|}{\sqrt{r_{\rm{Z}}^2-(\|\mathbf  r_{i}\|-x_0)^2}}\Big)
\TwoOneColumnAlternate{\notag\\&&}{}
+(x_0-\|\mathbf  r_{i}\|)
\sqrt{-\|\mathbf  r_{i}\|^2 +2\|\mathbf  r_{i}\|x_0-x_0^2+r_{\rm{Z}}^2}
\notag \\&& - 
r_{\rm{A}}^2 \tan^{-1}\Big(\tfrac{x_0}{\sqrt{r_{\rm{A}}^2-x_0^2}}\Big)
-x \sqrt{r_{\rm{A}}^2-x_0^2}
\IEEEeqnarraynumspace
\end{IEEEeqnarray}
and $x_0=\frac{r_{\rm{A}}^2+\|\mathbf  r_{i}\|^2-r_{\rm{Z}}^2}{2\|\mathbf  r_{i}\|}$. 
The terms in the denominator are given by
\begin{IEEEeqnarray}{rCl}\label{e: I1}
\bar J_1^i
&=&
\sum\limits_{ \substack{\ell\in \mathcal N_0 \\ 
\ell:\| \mathbf r_{i}-\mathbf r_{\ell} \| > 
\sqrt{\tfrac{\alpha-2}{\alpha\pi \lambda p_{\tx}}}}}
p_{\mathrm{tx}}\rho\cdot \| \mathbf r_{i}-\mathbf r_{\ell} \|^{-\alpha},
\vspace{0.5cm}
\\ \label{e: J_2^i general}
\bar J_2^i
&=&
\frac{2(\pi-\theta_s)
\rho\lambda p_{\mathrm{tx}} r_{\rm{Z}}^{2-\alpha}}{\alpha-2}
\TwoOneColumnAlternate{\\ \notag &&}{} 
+\frac{
\rho\lambda p_{\mathrm{tx}} }{\alpha-2}\Big[
\int_{\theta_s}^{2\pi-\theta_s}
\Big(-\|\mathbf  r_{i}\|\cos(\theta)
\TwoOneColumnAlternate{\\ \notag &&}{} 
+\sqrt{r_{\rm{A}}^2 - \|\mathbf  r_{i}\|^2\sin^2(\theta)}
\Big)^{2-\alpha}
d\theta\Big]
,
\IEEEeqnarraynumspace
\end{IEEEeqnarray}
where \begin{IEEEeqnarray}{rCl}\label{e: theta_s lemma}
\theta_s
&=&
\begin{cases}
0, & r_{\rm{Z}}+\|\mathbf r_{i}\|\le r_{\rm{A}} \\
\cos^{-1}\big(\frac{r_{\rm{A}}^2-r_{\rm{Z}}^2-\|\mathbf  r_{i}\|^2}{2r_{\rm{Z}}\|\mathbf  r_{i}\|}\big), & r_{\rm{Z}}+\|\mathbf r_{i}\|>r_{\rm{A}} \\
\end{cases}
.
\IEEEeqnarraynumspace
\end{IEEEeqnarray}
 It may be noted that \eqref{e: J_2^i general} has closed form expressions for any integer values of $\alpha$ larger than $2$ (detailed expressions  for $\alpha=3$ and $\alpha=4$ are given in 
\eqref{e: a=3}
and
\eqref{e: a=4}
 in   Appendix \ref{Appendix: calculation of aggregate intrf}).
\end{lemma}
\begin{proof}
See Appendix \ref{Appendix: proof of lemma BO}.
\end{proof}
The routing function that  optimizes  \eqref{e: lemma G(f,K) LB}  is termed  Bound Optimal routing (BO) here. This function is given by    
\begin{IEEEeqnarray}{rCl}\label{e: BO metric}
f_{\rm{BO}}(\mathcal M_0&&)
=
\argmax_{\substack{
i\in \mathcal N_0 }}
m_{_{\rm{BO}}}(i,\mathcal M_0)
\end{IEEEeqnarray}
where
the routing metric of the BO  is\begin{IEEEeqnarray}{rCl}\label{e: BO routing metric}
m_{_{\rm{BO}}}(i,\mathcal M_0)
\triangleq
p_{\rm{Z}}(i,\mathcal M_0) 
r_{i,0} \log_2
\left(1+\frac{ S_{i,0}}{\bar J_1^i+\bar J_2^i+\sigma_v^2}
\right).
\IEEEeqnarraynumspace
\end{IEEEeqnarray}
This metric  can be written in a closed form for  $\alpha =
3, 4, 5, ...$ (when \eqref{e: BO routing metric} has a closed form expression). 
However, this can hardly be called a `low-complexity' algorithm (recall  that \eqref{e: BO routing metric} requires the substitution of  
\eqref{e: p_Z}-\eqref{e: J_2^i general}).
   
\subsection{Narrow Knowledge Statistically-Optimal (NSO) Routing}
In the following, we derive a sub-optimal method which is much simpler to evaluate. This method produces nearly the same performance as the SO method, but can be evaluated using a simple, $1$-dimensional lookup table.
\par
As each node has only partial knowledge of the network state, it needs to perform the best routing decision given that knowledge. For this purpose, we use the statistical model of the unknown nodes. The statistical modeling cannot compensate for the unknown data. Nevertheless, it can give us the optimal balancing of the known data.
\par  
  To reduce the computational complexity, in the following we suggest using only  part of the available knowledge (which will be termed   \textit{narrow knowledge} here) when evaluating the routing metric for a specific node. The narrow knowledge of node $j$ on node $i$ is   $\mathcal M_j^i=\{r_{i,j},h_{i,j}\}$; i.e., taking only into account the distance to node $i$ and its channel gain (and ignoring the known data on all other neighbors). As we will see, the loss of performance due to the use of narrow knowledge is negligible, whereas the complexity is reduced significantly. 

This proposed low-complexity scheme will be denoted as Narrow knowledge Statistically Optimal (NSO). The NSO routing function  evaluates $G(i,\mathcal M_0)$ by solely using the available knowledge of the tested relay node, $i$. In mathematical terms, we  substitute $G(i,\mathcal M_0)$ by: :
\begin{IEEEeqnarray}{rCl}\label{e: G(f,K) LC}
G(i,\mathcal M_0) &\approx& G_\mathrm{N}(i,\mathcal M_0) \TwoOneColumnAlternate{\notag \\&\triangleq&}{\triangleq}  \mathbb E_{J|\mathcal M_0^i}\left\{r_{i,0} \log_2\left(1+\frac{  S_{i,0}}{J_{i,0}+\sigma_v^2}\right)\Big|
\mathcal M_0^i\right\}.
\IEEEeqnarraynumspace
\end{IEEEeqnarray}
As $r_{i,0}$ is known given $\mathcal{M}_0^i$, the optimization of this approximation results in the  NSO routing function: 
\begin{IEEEeqnarray}{rCl}\label{e: SO metric narrow knowledge SISO}
&&f_{\rm{NSO}}(\mathcal M_0)
=
\argmax_{\substack{
i\in \mathcal N_0 }}
m_{_{\rm{NSO}}}(i,\mathcal M_0^i)
\IEEEeqnarraynumspace
\end{IEEEeqnarray}
where
the routing metric of the NSO  is
\begin{IEEEeqnarray}{rCl}
\label{e: NBO routing metric ex}
m_{_{\rm{NSO}}}(i,\mathcal M_0^i)
\triangleq
r_{i,0} \cdot
\mathbb E\left\{\log_2\left(1+\frac{S_{i,0}}{J_{i,0}+\sigma_v^2} \right)
\Big|\mathcal M_0^i
\right\}.
\IEEEeqnarraynumspace
\end{IEEEeqnarray}
\par
Thus, the evaluation of 
(\ref{e: NBO routing metric ex}) is much simpler than the evaluation of 
\eqref{e: SO routing metric}, due to the difference in the conditional distribution of the interference, $J_{i,0}$. In the expectation in 
\eqref{e: NBO routing metric ex}, $J_{i,0}$ is independent of  narrow knowledge, and hence, its distribution is identical for all nodes.  
 The complexity of the optimization problem is not large, since the number of points for which the routing needs to calculate their metrics  is typically not large. To further simplify the evaluation of the metric in \eqref{e: NBO routing metric ex}, we suggest using the function $q(x_i)\triangleq\mathbb E\Big\{\log_2(1+\frac{x_i}{J_{i,0}+\sigma_v^2})\Big|\mathcal M^i_0\Big\}$.
Thus, this expectation is only a function of $S_{i,0}$ which can be written as 
\begin{IEEEeqnarray}{rCl}\label{e: NSO routing metric}
m_{_{\rm{NSO}}}(i,\mathcal M_0^i)
=
r_{i,0} \cdot
q(S_{i,0})
\end{IEEEeqnarray}
and the function $q(\cdot)$ can be evaluated as follows. 
\begin{figure}[t]
       \includegraphics[width=\FigureWidth]{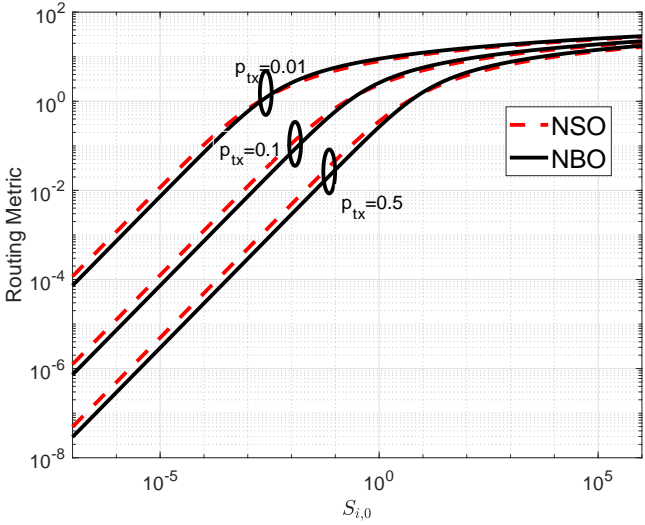}
    \caption{ NBO routing metric and NSO routing metric   vs.  $S_{i,0}$, for various $p_{\mathrm{tx}}\in\{0.01,0.1,0.5\}$ where $r_{i,0}=1$ and $\alpha=4$.}
    \label{fig: SISO NBO NSO}
\end{figure}
\par
For simple evaluation, the function $q(\cdot)$, which is the only complicated part in the evaluation of the routing metric, can be evaluated once and stored in
a lookup table. Thus, the complicated expectation can be evaluated only once, during the network design stage, and the routing metric can be  calculated easily in real time using the lookup table. Another alternative would be to evaluate the lookup table on the fly using the interference measurements at the transmitting node (again using the narrow knowledge assumption such that the interference statistics is identical for all nodes). 
\par
In the next subsection we present
an even simpler routing function for which we can give a simple, closed-form expression of the routing function that does not even require a lookup table.
\subsection{Narrow Knowledge Bound-Optimal (NBO) Routing}

  The equivalent of Lemma \ref{Theorem: lemma G(i,M) lower bound}, in \eqref{e: G(f,K) LC},  becomes:
\begin{IEEEeqnarray}{rCl}\label{e: G(i,M) LB NBO}
G_\mathrm{N}(i,\mathcal M_0) \ge
p_{\rm{Z}}(i,\mathcal M_0^i)r_{i,0}
\log_2\left(1+\frac{S_{i,0}}{\gamma+\sigma_v^2} \right)
\end{IEEEeqnarray}
with 
\begin{IEEEeqnarray}{rCl}\label{e: name of the equation}
\gamma=\rho\cdot\frac{2}{\alpha}
\left(\frac{\alpha\pi \lambda p_{\mathrm{tx}}\Gamma(1+{\frac{2}{\alpha}})}{\alpha-2}
\right)^{\frac{\alpha}{2}}
\end{IEEEeqnarray}
which is 
a constant that  depends on the network parameters. Noting that $p_{\rm{Z}}(i,\mathcal M_0^i)$ is identical for all nodes, the NBO routing function that optimizes \eqref{e: G(i,M) LB NBO} is given by 
\begin{IEEEeqnarray}{rCl}\label{e: NBO routing metric}
f_{\rm{NBO}}(\mathcal M_0)
&&=
\argmax_{\substack{
i\in \mathcal N_0 }}
m_{_{\rm{NBO}}}(i,\mathcal M_0^i)
\end{IEEEeqnarray}
where
the routing metric of the NBO  is
\begin{IEEEeqnarray}{rCl}\label{e: NBO routing metric exp}
m_{_{\rm{NBO}}}(i,\mathcal M_0^i)
\triangleq
r_{i,0}
\log_2\left(1+\gamma_b \cdot S_{i,0} \right)
\end{IEEEeqnarray}
and $\gamma_b=1/(\sigma_v^2+\gamma)$. 
In addition, this scheme coincides with the scheme proposed in \cite{Richter2014_IEEEI_SO}. While this scheme is not optimal as suggested in \cite{Richter2014_IEEEI_SO}, we will show in the following section that its performance  is quite close to that of the optimal scheme. Since the evaluation of the NBO routing metric is straightforward, we believe that this is indeed a good routing approach for practical networks.

 
\section{Numerical Results}\label{Section: Numerical Results}

In this section, we present simulation results that demonstrate
the efficiency of the proposed routing schemes, where the next hop is selected according to the SO, the  BO, the NSO  or the NBO routing functions in \eqref{e: opt metric SISO first}, \eqref{e: BO metric}, 
\eqref{e: SO metric narrow knowledge SISO} and \eqref{e: NBO routing metric}, respectively. In all the simulations we take the number of nodes to have a Poisson distribution with an average of $N_{\rm{nodes}}=300$. The nodes  are uniformly distributed in a disk with an area of size $\lambda/N_{\rm{nodes}}$, centered at the probe transmitting node. We also use the bias correction in \cite{Richter2014_PPP_analysis}. To simulate the common interference limited regime, we set $\sigma_v^2=0$. {} In this case, the transmitted power, $\rho$, and the node density, $\lambda$, have no effect on performance, and we set $\rho=1$ and $\lambda=1$.


\begin{figure}[t]
    \includegraphics[width=\FigureWidth]{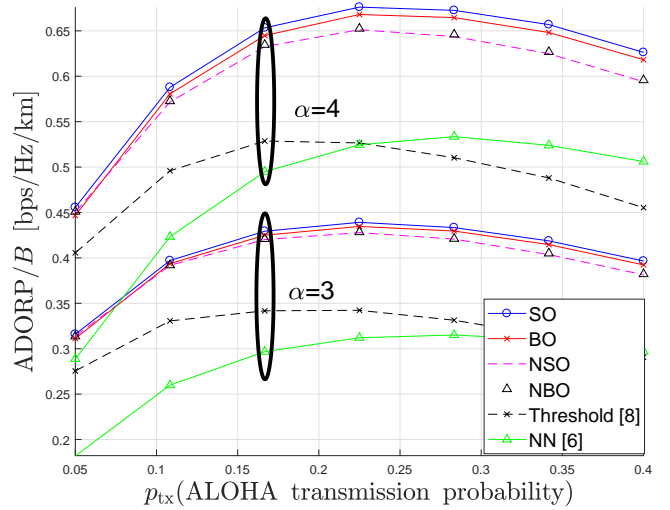}
    \caption{Normalized ADORP as a function of the ALOHA transmission probability for various routing
schemes,  $\bar N_A=30$ and $\alpha=3,4$.  }
    \label{fig: ADORP SISO alpha 3}
\end{figure}

To gain some insights into the characteristics of these routing schemes, we start by plotting the routing metrics for the two simpler schemes. Fig. \ref{fig: SISO NBO NSO} illustrates the NSO  and  the NBO routing metrics for $\alpha=4$. These schemes are based on narrow knowledge; hence, their routing metric for each node is solely a function of the distance and channel gain to this node. To further simplify the scenario, we consider a tested relay at a distance of  $r_{i,0}=1$, such that the metrics are only a function of the channel gain.   In the NSO  scheme, \eqref{e: NSO routing metric}, the resulting metric is the rate:  $q(S_{i,0})=\mathbb E\left\{\log_2\left(1+\frac{S_{i,0}}{J_{i,0}+\sigma_v^2} \right)
\Big|\mathcal M_0^i
\right\}
$, and the  aggregate interference can be evaluated by a  Monte-Carlo simulation. In the NBO  scheme, \eqref{e: NBO routing metric exp}, the resulting metric is the rate:  
$\log_2\left(1+\gamma_b \cdot S_{i,0} \right)$ (which is obviously much easier to evaluate than the NSO metric). 
\begin{figure}[t]
     \includegraphics[width=\FigureWidth]{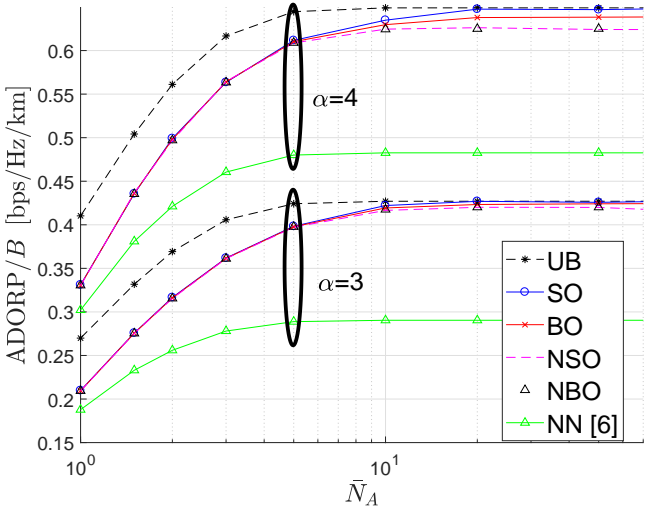}
    \caption{Normalized ADORP as a function of the average number of nodes within the routing zone, 
    $\bar N_A$,  for various routing
schemes,  $p_{\tx}=0.15$ and $\alpha=3,4$.  }
    \label{fig: ADORP_vs_rA}
\end{figure}

\par
The
difference between these two metrics is minimal; thus we expect their performance to be quite similar.
As expected, both metrics are very sensitive to changes in the channel gain when the channel gain is low. On the other hand, at very high channel gains the achievable rate increases very slowly, and hence the routing metrics are much less sensitive.
\par
The simpler NBO metric considers the interference as a constant, and  is less efficient than the NSO. 
Compared to the NSO metric,  the NBO metric gives higher weights to nodes with higher channel gains, $S_{i,0}$, and gives lower weights to nodes with lower channel gains. This difference leads to the small performance gap between NBO and NSO routing, as can be seen in Fig. \ref{fig: ADORP SISO alpha 3}. 
\par

 Fig. \ref{fig: ADORP SISO alpha 3} depicts the normalized ADORP (ADORP/$B$) as a function of the
ALOHA transmission probability, $p_{\mathrm{tx}}$, for a system with a path
loss exponent of $3$ and $4$.  All of the proposed schemes perform quite similarly, and achieve the maximum throughput near $p_{\tx}=0.2$. The SO curve  serves as an upper bound on  the achievable performance, based only on local knowledge. The SO routing function is  given in \eqref{e: opt metric SISO first} and the aggregate interference is evaluated by  Monte-Carlo simulations.  
\par
The BO scheme, \eqref{e: p_Z}-\eqref{e: BO metric}, enables a closed form evaluation of the routing metric, at only negligible loss of performance.
($1.2$\% at $\alpha=4$ and $1$\% at $\alpha=3$).
 However, its evaluation still requires high computational complexity.  
 On the other hand, the narrow knowledge schemes (NBO and NSO) are much simpler to evaluate although they incur a slightly larger performance gap ($3.8$\% at $\alpha=4$ and $2.6$\% at $\alpha=3$).    
The NBO is even simpler  to evaluate and its performance loss compared to the NSO curve is negligible. 
\par
For comparison, Fig. \ref{fig: ADORP SISO alpha 3} also depicts the performance of two geographic routing schemes from the literature: 
Nearest-Neighbor (NN) routing (e.g., \cite{haenggi2005routing,nardelli2012efficiency}), and Threshold routing \cite{weber2008longest}. 
The NN scheme
is ranked high among geographic routing schemes. Nevertheless it is inferior to all of our novel schemes.  It may be noted that  most of the  performance gain in our schemes comes from  the optimal combination of  geographic and channel state knowledge. 
\par
The Threshold routing scheme (which was the first published opportunistic relaying scheme) also takes advantage of the
CSI to outperform the most-progress-in-radius routing scheme
(e.g., \cite{nardelli2012efficiency}). This scheme is an adaptation of Weber et al.
\cite{weber2008longest} to the present setup. In this scheme, the transmitting
node first identifies the set of relays for which the power
of the desired signal is above a certain threshold\footnote{It may be noted that \cite{weber2008longest} used a comparison of the actual SINR to the threshold. But in our setup, the transmitting node does not know the interference power. Hence, the
transmitting node chooses the next relay (out of these relays)
that maximizes progress towards the destination. The threshold
value was chosen to maximize the throughout in each scenario.
However, the results show that the use of the CSI in the
threshold scheme is far from optimal, and this scheme is also
inferior to our novel schemes.}.
{{TABLE \ref{Table: Routing Characteristic} summarises the presented routing schemes and the differences among them.}}
\TwoOneColumnAlternate
{\begin{table} 
\begin{center} 
\begin{tabular}{ | p{3.1cm} | p{1.5cm} |  p{0.95cm} |  p{1.4cm}|} 
\hline 
\label{table: opt beta}
\uline{Routing Scheme} & \uline{Computational Complexity} 
& \uline{Use CSI} & \uline{Performance}  \\\hline
Statistically-Optimal (SO)   & High & Partial & Optimal \\\hline
Bound-Optimal (BO)  & Medium & Partial & Near optimal\\\hline
Narrow Knowledge Statistically-Optimal (NSO)    & Medium & Narrow & Near optimal \\\hline
Narrow Knowledge Bound-Optimal (NBO)   & Low & Narrow & Near optimal\\\hline
Threshold \cite{weber2008longest}   & Low & Narrow & Medium\\\hline
Nearest-Neighbor (NN) \cite{nardelli2012efficiency}& Low & Locations & Low \\ \hline
\end{tabular} 
\caption{Major characteristics of the compared routing schemes.}
\label{Table: Routing Characteristic}
\end{center}
\end{table}
} 
{\begin{table} 
\begin{center} 
\begin{tabular}{ | p{5cm} | p{3cm} |  p{1.8cm} |  p{2.3cm}|} 
\hline 
\label{table: opt beta}
\uline{Routing Scheme} & \uline{Computational Complexity} 
& \uline{Use CSI} & \uline{Performance}  \\\hline
Statistically-Optimal (SO)   & High & Partial & Optimal \\\hline
Bound-Optimal (BO)  & Medium & Partial & Near optimal\\\hline
Narrow Knowledge Statistically-Optimal (NSO)    & Medium & Narrow & Near optimal \\\hline
Narrow Knowledge Bound-Optimal (NBO)   & Low & Narrow & Near optimal\\\hline
Threshold \cite{weber2008longest}   & Low & Narrow & Medium\\\hline
Nearest-Neighbor (NN) \cite{nardelli2012efficiency}& Low & Locations & Low \\ \hline
\end{tabular} 
\caption{Major characteristics of the compared routing schemes.}
\label{Table: Routing Characteristic}
\end{center}
\end{table}
} 


\par
The gain from the use of local knowledge obviously depends on the radius of the routing zone, $r_A$. In order to use a scale-free variable, it is convenient to characterize the routing zone area in terms of the   average number of nodes within the routing zone,  $\bar N_A\triangleq\lambda\pi r_A^2$.  Fig. \ref{fig: ADORP SISO alpha 3} presents the  case  of $\bar N_A=30$.
Fig. \ref{fig: ADORP_vs_rA} presents the normalized ADORP as a function of the   average number of nodes within the routing zone for  $p_{\tx}=0.15$ and $\alpha=3,4$. 
\par
Fig. \ref{fig: ADORP_vs_rA} shows that  the performance of all the proposed routing schemes exhibits no  significant loss if  $\bar N_A$ is greater than $10$.
 For smaller routing zones, the gain decreases significantly. However, our schemes outperform the NN for any radius of the routing zone.  
\par
It is also crucial to  consider the probability that the routing zone is empty. In this case, none of the routing schemes will be able to transfer data, which will inevitably lead to a decrease in performance. To demonstrate this effect, 
 Fig. \ref{fig: ADORP_vs_rA} also depicts  an upper bound (UB) on the achievable performance. The UB curve is calculated as the maximum normalized ADORP of the SO curve multiplied by the probability to have a non-empty routing zone, $\mathbb P(A_R\text{ has at least single node})=1-\exp(-\bar N_A)$. 
\par
The upper bound helps us to differentiate between the two effects of increasing the size of the routing zone: better selection and better processing. For small routing zone sizes, each transmitting node has only a few nodes to select from. Hence, any increase in the size of this area can add a better relay and hence improve  performance. This effect is saturated around $\bar N_A=5$, and above that, the only gain is from the better prediction of the interference power. Furthermore, this prediction is generated solely by the SO and BO schemes, so they are the only ones to gain at high $\bar N_A$. 

\begin{figure}[t]
    \includegraphics[width=\FigureWidth]{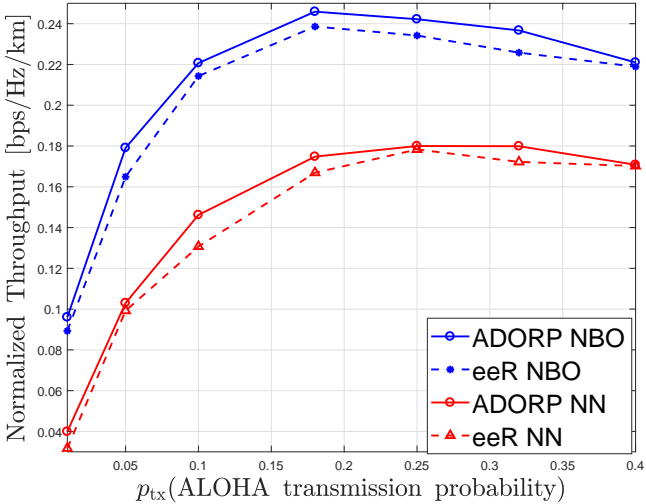}
    \caption{The density of the normalized end-to-end rate distance metric (eeR)  as a function of the ALOHA transmission probability in a full network simulation. The figure depicts the performance of the NBO and the NN routing
schemes, for a path-loss exponent of  $\alpha=3$. }
    \label{fig: E2E sim}
\end{figure}
To further demonstrate the advantage of the novel routing schemes, we also performed  a complete network simulation, in which messages are routed from sources to destinations according to the mechanism described in Subsection \ref{subsection: routing mechanism}. 
{{The simulation was implemented in Matlab, and the results are based on 1000 network realizations.}}     
The simulation includes an average of $100$ nodes uniformly distributed over a simulation area of $1000$m$^2$. 
We assume that each node has a home location and an  i.i.d. mobility  model (e.g., \cite{gong2013local}). 
The mobility follows a symmetric normal distribution with a variance of $2.84$m.  This mobility rate represents a low mobility that
is only sufficient to unravel the network bottlenecks when choosing a simulation length that depends on $p_{\tx}$ according to $t_s=10^5/p_{\tx}+8\cdot 10^5$  slots. Thus,
each simulation was run for a duration   of $t_s\cdot  T$ seconds, where  $T$ was  the duration of a time slot. Each message contained  $20BT$ bits where $B$ was the bandwidth. The messages were not subjected to any delay constraint.  
The performance was measured by the normalized density of end-to-end rate distance metric (eeR) \cite{Richter2014_IEEEI_SO}, where we summed the distance-bits product for all successful messages, and divided by the size of the area-time and by the bandwidth. The normalized eeR metric is given by \cite[Eq. (9)]{Richter2014_IEEEI_SO}:
\begin{IEEEeqnarray}{rCl}\label{e: eq name}
\mathrm{eeR}\triangleq \frac{1}{ T_{\mathrm{T}} A B}\sum_\ell L_\ell K \cdot i_\ell
\end{IEEEeqnarray}
where $A$ is the simulation area, $T_\mathrm{T}=t_s\cdot  T$ is the simulation time, $K=20BT$ is the number of bits per message, $i_\ell$ is the successful delivery indicator and $L_\ell$ is the distance between source and destination at the time of message generation. 
Using the ergodic rate approach, a message is assumed to be successfully decoded if it accumulates a mutual information that is equal or larger than $K$ in each of its hops.
\par
Fig. \ref{fig: E2E sim} depicts the normalized eeR as a function of the
ALOHA transmission probability, $p_{\mathrm{tx}}$, for a system with a path
loss exponent of $3$. 
  The figure shows the eeR performance of the NN
scheme and the NBO scheme. The figure also shows the relevant normalized ADORP for each scheme. As can be  seen,  the end-to-end performance indeed converges to the normalized  ADORP. Also, as expected from the previous
results, the NBO scheme significantly outperforms the NN
scheme. The gains vary between $30\%$ to $180\%$.
\par
 Fig. \ref{fig: ADORP vs SNR} depicts the normalized ADORP  as a function of the transmission power for various routing
schemes,  $p_{\tx}=0.15$ and $\alpha=3$. The transmission power is depicted through $\overline{\mathrm{SNR}}_1=\rho/\sigma_v^2$, the averaged SNR for a receiver at $1$ unit distance. As expected, for small SNRs, the normalized ADORP increases linearly  with $\overline{\mathrm{SNR}}_1$. For $\overline{\mathrm{SNR}}_1$ above $-20$dB, the noise becomes negligible, and the networks operate at the interference limited regime. For all SNR values, the proposed schemes outperforms the NN  scheme. 
For very small SNR, the NN becomes close to optimal. 

\begin{figure}[t]
        \includegraphics[width=\FigureWidth]{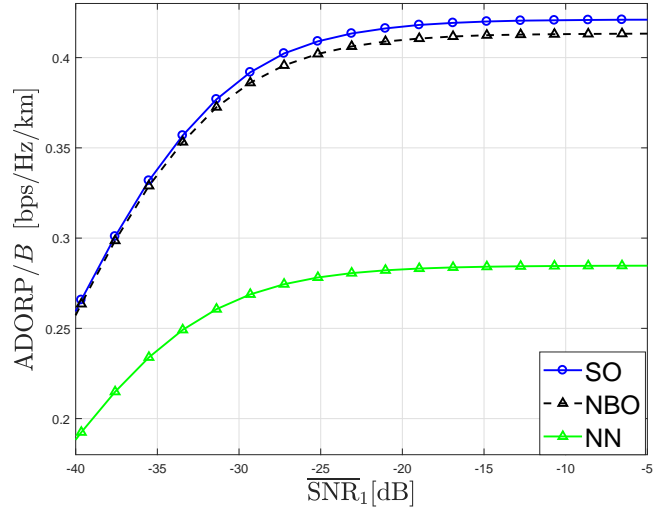}
    \caption{Normalized ADORP as a function of the transmission power, for various routing
schemes,  $p_{\tx}=0.15$ and $\alpha=3$. The transmission power is depicted through $\overline{\mathrm{SNR}}_1=\rho/\sigma_v^2$, the averaged SNR for a receiver at $1$ unit distance.}
    \label{fig: ADORP vs SNR}
\end{figure}
\section{Conclusions\label{Section: Conclusions}}
In this paper, we proposed novel routing schemes for
WANETs, where the node locations are represented by a  PPP. We focus on geographic routing schemes, where  the routing decision at each node is based solely on local knowledge on the nodes  within its routing zone (geographical locations
and channel gains).  
\par
\par 
In this setup, we were able to formulate the routing objectives as an optimization problem, and to present the routing scheme that solves this optimization problem. The resulting routing scheme requires high computational complexity. Nevertheless, the presentation of the optimal routing scheme is important as it gives an upper bound on the performance of any other routing scheme. 
\par
The knowledge of the maximal routing performance also allowed us  to  present and evaluate  three
suboptimal low-complexity routing schemes. These schemes were derived by using only 
part of the available knowledge, by taking a simple lower bound or by combining these two methods. The performance of all three schemes is very close to  the performance of the optimal scheme (and in particular
at low transmission probability, where the performance
gap tends toward zero). Furthermore, the performance of all the proposed routing
schemes outperforms the performance of previously published routing schemes.
\par

This paper considers the fast mobility model. 
While this model is frequently used in stochastic geometry analysis, most practical networks have much tighter delay constraints. Future work should consider a more realistic scenario where messages must be transferred before the network topology changes. 
In this case, some message buffers in the network may be empty, and the network performance is dominated by its bottlenecks.  
Hence, the choice of relay must take  the state of the buffer, the desired directions of messages in the buffer and the load in each of the potential relays into account. 
In particular, one can consider an improved routing scheme in which  the routing metric suggested above is modified to give proper weights to all of these factors.  Such a metric might be able to trade off between the optimization of network throughput and the minimization of the message delay.
Future work  could also  consider   the use of multiple antennas at each node.   


 


\begin{appendices} 
\section{Proof of Lemma \ref{Theorem: lemma G(i,M) lower bound}
 \label{Appendix: proof of lemma BO}}
The lower bound is based on the use of  the law of total expectation, while conditioning on the event that the threshold zone (with radius of $r_{\rm{Z}}$) is free of interferers. 
Denote by $d_i $  the distance between the tested relay (node $i$)  and its nearest interfering node.  Equation 
\eqref{e: G(i,k)} can be written as: 
\begin{IEEEeqnarray}{rCl}\label{e: G(f,K) LB}
 \TwoOneColumnAlternate{&G(&i,\mathcal M_0)  }
 {G(i,\mathcal M_0)} 
 \TwoOneColumnAlternate{\\&=&}{&=&} 
(1-p_{\rm{Z}}(i,\mathcal M_0)) \mathbb E\Big\{
r_{i,0} \log_2
(1+\tfrac{S_{i,0}}{J_{i,0}+\sigma_v^2}
)\Big|d_i\le r_{\rm{Z}},\mathcal M_0\Big\}
\TwoOneColumnAlternate{\notag \\}{\\} 
&&+p_{\rm{Z}}(i,\mathcal M_0) \mathbb E\Big\{
r_{i,0} \log_2
\Big(1+\tfrac{ S_{i,0}}{J_{i,0}+\sigma_v^2}
\Big)\Big|d_i> r_{\rm{Z}},\mathcal M_0\Big\}
\notag \\
&\ge&%
p_{\rm{Z}}(i,\mathcal M_0) \mathbb E\Big\{
r_{i,0} \log_2
\Big(1+\tfrac{ S_{i,0}}{J_{i,0}+\sigma_v^2}
\Big)\Big|d_i> r_{\rm{Z}},\mathcal M_0\Big\}
\notag \\
&\ge&
p_{\rm{Z}} (i,\mathcal M_0)
 \cdot r_{i,0} \log_2
\left(1+\frac{S_{i,0}}{\mathbb E\{J_{i,0}|d_i> r_{\rm{Z}},\mathcal M_0\}+\sigma_v^2}
\right) \notag
\end{IEEEeqnarray}
where  $ p_{\rm{Z}}(i,\mathcal M_0)\triangleq \mathbb P(d_i>r_{\rm{Z}}\big|\mathcal M_0)$ is evaluated in Appendix \ref{Appendix: p_Z}, and the last line uses the Jensen inequality. This bound holds for any choice of the threshold radius, $r_{\rm{Z}}$. A useful value for this radius was shown to be \cite{Richter2014_PPP_analysis}: $r_{\rm{Z}}=\sqrt{\tfrac{\alpha-2}{\alpha \pi p_{\mathrm{tx}} \lambda}}$. The rest of the proof focuses on the evaluation of the conditional expectation of the aggregate interference.

We denote by  $J_{i,0}(A)$  the aggregate interference at the $i$-th node which is induced only from the transmitting nodes within  area $A$.   We consider the division of the plane into the non-overlapping areas  $A_{\mathrm{Z},i}$, $(\overline{A_{\mathrm{Z},i}} \cap A_{\mathrm{R}})$ and 
$(\overline{A_{\mathrm{Z},i}} \cap\overline{A_{\mathrm{R}}})$, where  $\overline {A}$ denotes the area outside $A$ (see Fig. \ref{fig: Routing Zone General} for an illustration). As we condition on $d_i> r_{\rm{Z}}$, we have:
\begin{IEEEeqnarray}{rCl}\label{e: I1+I2+I3}
\mathbb E\{J_{i,0}|d_i> r_{\rm{Z}},\mathcal M_0\}
&=& 
\mathbb E\{J_{i,0}(\overline{A_{\mathrm{Z},i}} \cap A_{\mathrm{R}})|\mathcal M_0\}
\\ \notag &&+
\mathbb E\{J_{i,0}(\overline{A_{\mathrm{Z},i}} \cap\overline{A_{\mathrm{R}}} )|\mathcal M_0\} .
\end{IEEEeqnarray}
Thus, we define
\begin{IEEEeqnarray}{rCl}
\bar J_1^i&\triangleq & \mathbb E\left\{J_{i,0}(\overline{A_{\rm{Z},i}} \cap A_{\mathrm{R}})|\mathcal M_0 \right\}
\label{I_1} \\
\bar J_2^i&\triangleq&\mathbb E\left\{ J_{i,0}(\overline{A_{\rm{Z},i}} \cap \overline{A_{\mathrm{R}}})|\mathcal M_0 \right\}
\label{I_2} 
\end{IEEEeqnarray}
and evaluate each expectation separately.

Writing the expectation in \eqref{I_1} more explicitly,  we consider $J_1^i$ to be the aggregate interference at the next relay, which is induced by  the transmitting nodes within $(\overline{A_{\rm{Z},i}} \cap A_{\mathrm{R}})$: 
\begin{IEEEeqnarray}{rCl}\label{e: J1 origin}
&& J_1^i =
\sum\limits_{\substack{ \ell\in \mathcal N_0\\ 
\ell:\| \mathbf r_{i}-\mathbf r_{\ell} \| > 
\sqrt{\tfrac{\alpha-2}{\alpha\pi \lambda p_{\mathrm{tx}}}}}}
g_\ell\rho\cdot 
r_{i,\ell}^{-\alpha}
W_{i,\ell},
\IEEEeqnarraynumspace
\end{IEEEeqnarray}
where $g_\ell$ is an indicator function which equals  $1$ if the $\ell$-th node is scheduled to  serve as a transmitting node, and $\mathbb E\{g_\ell\}=p_{\tx}$.
The expectation of \eqref{e: J1 origin}, $\bar J_1^i$, can be evaluated by using the statistical independence of transmission decisions:
\begin{IEEEeqnarray}{rCl}\label{e: J1 re origin}
\bar J_1^i
&=&
\sum\limits_{\substack{ \ell\in \mathcal N_0\\ 
\ell:\| \mathbf r_{i}-\mathbf r_{\ell} \| > 
\sqrt{\tfrac{\alpha-2}{\alpha\pi \lambda p_{\mathrm{tx}}}}}}
p_{\mathrm{tx}}\rho\cdot \| \mathbf r_{i}-\mathbf r_{\ell} \|^{-\alpha} .
\end{IEEEeqnarray}

The evaluation of   $\bar J_2^i$ requires a cumbersome analysis of the shape of the area which is seen by the $i$-th node, and its evaluation is given in Appendix \ref{Appendix: calculation of aggregate intrf}. 
   
\section{Evaluation of $p_{\rm{Z}}(i,\mathcal M_0)$\label{Appendix: p_Z}}
In this appendix we evaluate the probability that there is no interfering transmitting node within the threshold zone, $p_{\rm{Z}}(i,\mathcal M_0)$, which is given in \eqref{e: p_Z}.
Given $\mathcal M_0$, the probability 
$p_{\rm{Z}}(i,\mathcal M_0)$  depends on the number of nodes in the routing zone that lie in the threshold zone of  node $i$.  Recall that node $j$ lies in the threshold zone of node $i$ if 
$\|\mathbf r_j-\mathbf r_i \|\le r_{\rm{Z}}$.
The evaluation of $p_{\rm{Z}}(i,\mathcal M_0)$ considers two cases:

\subsubsection
{$A_{\mathrm{Z},i}$ is located completely within $A_{\mathrm{R}}$}
In this case, the probe knows the locations of all nodes within the threshold zone, and can compute the  number of nodes within the threshold zone:  \begin{IEEEeqnarray}{rCl}\label{e:}
N_{{\rm{Z}},i}
=
\sum_{\ell\in \mathcal N_0}\mathbf 1_{\{\|\mathbf r_i-\mathbf r_\ell \|\le r_{\rm{Z}} \}}
.
\end{IEEEeqnarray}
The probability that all nodes are receiving nodes is given by $(1-p_{\mathrm{tx}})^{N_{{\rm{Z}},i}}$,
 which is  the first case in \eqref{e: p_Z}.
\subsubsection
{Part of $A_{\mathrm{Z},i}$ is not located  within $A_{\mathrm{R}}$}
 In this case $A_{\mathrm{Z},i}\cap A_{\mathrm{R}}$ is not necessarily empty, and we need to calculate the probability that a transmitting node is located within  $A_{\mathrm{Z},i}$ but outside the routing zone,  $A_{\mathrm{R}}$. The area of the threshold zone outside the routing zone is $B_{{\rm{T}},i}\triangleq|\overline{A_{\mathrm{R}}}\cap A_{\mathrm{Z},i}|$. The local knowledge of the probe transmitting node contains no knowledge of the nodes in the area $B_{{\rm{T}},i}$, and we only know that they are characterized by a PPP with density  $p_{\mathrm{tx}}\lambda$. The size of $B_{{\rm{T}},i}$ is given in (\ref{e: B1 by integral}). Thus, the probability that this area has no  transmitting nodes is given by $e^{-\lambda p_{\mathrm{tx}} B_{{\rm T},i}}$. 
Multiplying also by the probability that there are no transmitting nodes  in $A_{\mathrm{R}}\cap A_{\mathrm{Z},i}$,
 results in the second case in \eqref{e: p_Z}.

\section{Evaluation of $\bar J_2^i$ 
\label{Appendix: calculation of aggregate intrf}}
In this appendix we evaluate $\bar J_2^i$, the average interference power from the area $\overline{A_{\rm{Z},i}} \cap \overline{A_{\mathrm{R}}}$ given $\mathcal M_0$. The final result of this appendix is given in \eqref{e: J_2^i general}.
For this calculation, it is convenient to set the axis system so that the origin is located at the tested next relay and the probe transmitting node is located at $x=-\|\mathbf  r_{i}\|$, $y=0$.
Thus, the circle that contains the routing zone, $A_{\mathrm{R}}$, is given by
\begin{IEEEeqnarray}{rCl}
\label{e: Circle probe transmitter}
(x+\|\mathbf  r_{i}\|)^2+y^2 &=& r_{\rm{A}}^2 .
\end{IEEEeqnarray}

Converting to polar coordinates ($x=r\cos (\theta)$ and $y=r\sin (\theta)$), we get 
\begin{IEEEeqnarray}{rCl} \label{e: Circle probe transmitter polar}
r^2+2 r\cos(\theta)\cdot \|\mathbf  r_{i}\|+\|\mathbf  r_{i}\|^2 &=& r_{\rm{A}}^2     .
\end{IEEEeqnarray}
Hence, the edge of routing zone can be described by
\begin{IEEEeqnarray}{rCl}\label{e: r_0}
r_0(\theta)\triangleq
-
\|\mathbf  r_{i}\|\cos(\theta)+\sqrt{r_{\rm{A}}^2 - \|\mathbf  r_{i}\|^2\sin^2(\theta)}
\end{IEEEeqnarray}
and the area can be evaluated by an integral within this area.

\begin{figure}[t]
    \center
    \includegraphics[scale=0.35,trim={1cm 4.cm 0cm 4.1cm},clip=true]{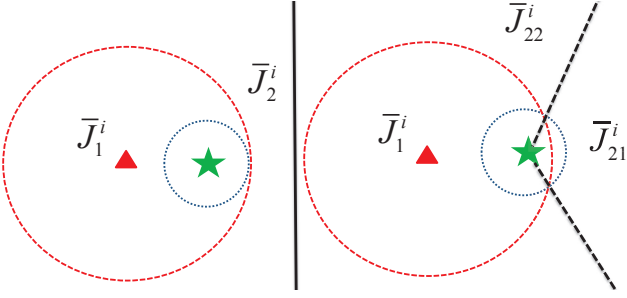}
    \caption{ The dashed circle represents the routing zone, $A_{\mathrm{R}}$, centered at the probe transmitting node.  The threshold zone of node $i$, $A_{\mathrm{Z},i}$, is marked by the small yellow dotted circle around the tested relay. We  evaluate $\bar J_1^i$ and $\bar J_2^i$ which are defined in \eqref{I_1} and \eqref{I_2}, respectively. We  distinguish between two cases: the left side considers  that  $A_{\mathrm{Z},i}$ is located completely within $A_{\mathrm{R}}$. The right side considers  that part of  $A_{\mathrm{Z},i}$ is located outside $A_{\mathrm{R}}$, and thus we use  $\bar J_2^i=\bar J_{21}^i+\bar J_{22}^i$. The cone width of $\bar J_{21}^i$ is $2\theta_s$. These two cases are used in Appendix \ref{Appendix: calculation of aggregate intrf} in the evaluation of $\bar J_1^i$ and $\bar J_2^i$.}
    \label{fig: J1_J2}
\end{figure}

The characteristic function of  $\bar J_2^i$ (the expectation of the aggregate interference at the next relay which is induced by  the transmitting nodes within 
$\overline{A_{\mathrm{Z},i}} \cap\overline{A_{\mathrm{R}}}$) is given by \cite{venkataraman2006shot}: 
\begin{IEEEeqnarray}{rCl}\label{e: eq name}
\Phi_{J_{i,0}}(s) 
\TwoOneColumnAlternate{&&\\ \notag=}{=}
\exp &\Big(&-
2\pi\lambda p_{\tx}
\mathbb E \Big\{
\int\int_{\delta(\theta)}^{\infty}
(1-e^{-s\rho \lambda  p_{\tx} W r^{-\alpha}}r dr d\theta)
\Big\}\Big) 
\end{IEEEeqnarray}
where $\delta(\theta)$ defines the edge of the interference free zone. The expectation over $J_{i,0}$ can be evaluated by substituting $s=0$ into the derivative of $\Phi_{J_{i,0}}(s)$ (see for example \cite{Richter2014_PPP_analysis}). However, we need to distinguish between two cases (See Fig. \ref{fig: J1_J2}): 
 in the first (left) case, the threshold zone is located completely within the routing zone (i.e.,  $r_{\rm{Z}}+\|\mathbf r_{i}\|\le r_{\rm{A}}$).  In the second (right) case,   part of the threshold zone is located outside  the routing zone (i.e., $r_{\rm{Z}}+\|\mathbf r_{i}\|>r_{\rm{A}}$). 

The first case is simpler, and is characterized by $\delta(\theta)=r_0(\theta)$, which results in
\begin{IEEEeqnarray}{rCl}\label{e: J_2 first case}
\bar J_2^i&=&\frac{\partial\Phi_{J_2^i}(s)}{\partial s} \Big|_{s=0}
=
\rho\lambda p_{\tx} 
\int_0^{2\pi}
\int_{r_0(\theta)}^{\infty}
r^{1-\alpha}
dr
d\theta
\notag \\
&=&
\frac{
\rho\lambda p_{\mathrm{tx}}}{2-\alpha}\Big[
\int_{0}^{2\pi}
\Big(
-\|\mathbf  r_{i}\|\cos(\theta)
\TwoOneColumnAlternate{\\ \notag &&}{} 
+\sqrt{r_{\rm{A}}^2 - \|\mathbf  r_{i}\|^2\sin^2(\theta)}
\Big)^{2-\alpha}
d\theta\Big].
\end{IEEEeqnarray}

In the second case, we start by computing the angle of  intersection points between the threshold zone (which is given by $r^2=r_{\rm{Z}}^2$) and the routing zone. 
Substituting in \eqref{e: Circle probe transmitter polar}, we get
\begin{IEEEeqnarray}{rCl}\label{e: intersection points}
\cos(\theta)&=&
\frac{r_{\rm{A}}^2-r_{\rm{Z}}^2-\|\mathbf  r_{i}\|^2}{2r_{\rm{Z}}\|\mathbf  r_{i}\|}.
 \IEEEeqnarraynumspace
\end{IEEEeqnarray}
Denoting the angle of  intersection point by $\theta_s\triangleq \cos^{-1}\big(\frac{r_{\rm{A}}^2-r_{\rm{Z}}^2-\|\mathbf  r_{i}\|^2}{2r_{\rm{Z}}\|\mathbf  r_{i}\|}\big)$
we get
\begin{IEEEeqnarray}{rCl}\label{e: eq name}
\delta(\theta)=\begin{cases}
R_Z, & |\theta|<\theta_s \\
r_0(\theta), & \mathrm{otherwise}
\end{cases}.
\end{IEEEeqnarray}

Thus,  
  we  divide the integral into two parts, $\bar J_2^i=\bar J_{21}^i+\bar J_{22}^i$, where
\begin{IEEEeqnarray}{rCl}\label{e: J_21}
\bar J_{21}^i&=&
\rho\lambda p_{\tx} 
\int_{-\theta_s}^{\theta_s}
\int_{r_{\rm{Z}}}^{\infty}
r^{1-\alpha}
dr
d\theta
\TwoOneColumnAlternate{\\ \notag &=&}{=} 
\frac{2\theta_s
\rho\lambda p_{\mathrm{tx}} r_{\rm{Z}}^{2-\alpha}}{\alpha-2},
\IEEEeqnarraynumspace
\end{IEEEeqnarray}
and 
\begin{IEEEeqnarray}{rCl}\label{e: J_22}
\bar J_{22}^i&=&
\rho\lambda p_{\mathrm{tx}}  
\int_{\theta_s}^{2\pi-\theta_s}
\int_{r_0(\theta)}^{\infty}
r^{1-\alpha}
dr
d\theta
\notag \\
&=&
\frac{
\rho\lambda p_{\mathrm{tx}} }{\alpha-2}\Big[
\int_{\theta_s}^{2\pi-\theta_s}
\Big(-\|\mathbf  r_{i}\|\cos(\theta)
\TwoOneColumnAlternate{\\ \notag &&}{} 
+\sqrt{r_{\rm{A}}^2 - \|\mathbf  r_{i}\|^2\sin^2(\theta)}
\Big)^{2-\alpha}
d\theta\Big].
\IEEEeqnarraynumspace
\end{IEEEeqnarray}
It may be noted that \eqref{e: J_2 first case} coincides with \eqref{e: J_22}  if we set  $\theta_s=0$. Thus, the two different cases can be jointly summarized by Equations \eqref{e: I1}-\eqref{e: theta_s lemma}.
\par
Unfortunately, the calculations of the integrals in \eqref{e: J_2 first case} and \eqref{e: J_22} only have closed form solutions   for integer values of  $\alpha> 2$. However,  in the general case, for other non-integer values where $\alpha>2$, these integrals can be evaluated  numerically.
For example, for  $\alpha=3$ we get
\begin{IEEEeqnarray}{rCl}\label{e: a=3}
\bar J_{22}^i
&=&
\rho\lambda p_{\tx} 
\left[-
\frac{2\|\mathbf  r_{i}\|\sin(\theta_s)}{r_{\rm{A}}^2-\|\mathbf  r_{i}\|^2}\right.
\\ &&+\Big(E(2\pi-\theta_s|\tfrac{\|\mathbf r_i\|^2}{r_{\rm{A}}^2})-E(\theta_s|\tfrac{\|\mathbf r_i\|^2}{r_{\rm{A}}^2})\Big) \notag
\TwoOneColumnAlternate{\\ \notag && }{} 
\cdot  \left.
\frac{r_{\rm{A}}^2 \sqrt{\frac{\|\mathbf r_i\|^2\cos(2\theta_s)}{r_{\rm{A}}^2}
-\frac{\|\mathbf r_i\|^2}{r_{\rm{A}}^2}+2}}
{(r_{\rm{A}}^2-\|\mathbf r_i\|^2)\sqrt{2r_{\rm{A}}^2+\|\mathbf r_i\|^2\cos(2\theta_s)-\|\mathbf r_i\|^2}}
\right]
\IEEEeqnarraynumspace
\end{IEEEeqnarray}
where  $E(\phi|k)\triangleq \int_0^{\phi}\sqrt{1 - k^2\sin^2(\theta)}d\theta$ is the Elliptic Integral of the second kind.
For $\alpha=4$ we get  
\begin{IEEEeqnarray}{rCl}\label{e: a=4}
\bar  J_{22}^i
&=&
\frac{\rho\lambda p_{\mathrm{tx}}}{4(r_{\rm{A}}^2-\|\mathbf r_i\|^2)^2}
\Big[
-2\|\mathbf  r_{i}\|\sin(\theta_s)
\TwoOneColumnAlternate{\\ \notag &&}{} 
\cdot \sqrt{4r_{\rm{A}}^2+2\|\mathbf r_i\|^2\cos(2\theta_s)-2\|\mathbf r_i\|^2}
\TwoOneColumnAlternate{\notag }{}
\\
&&-
2 r_{\rm{A}}^2\tan^{-1}\Big(\tfrac{\sqrt{2}\|\mathbf  r_{i}\|\sin(\theta_s)}{\sqrt{2r_{\rm{A}}^2+\|\mathbf r_i\|^2\cos(2\theta_s)-\|\mathbf r_i\|^2}}\Big)
\TwoOneColumnAlternate{\notag\\  &&}{} 
+4 r_{\rm{A}}^2(\pi-\theta_s)-2\|\mathbf r_i\|^2\sin(2\theta_s)
\Big]
. \notag 
\IEEEeqnarraynumspace
\end{IEEEeqnarray}

\end{appendices} 
\bibliographystyle{IEEEtran}    
\bibliography{Routing_SISO_bib}

\end{document}